\documentclass[oribibl]{llncs}
\usepackage{color}
\usepackage[algosection,lined,ruled,linesnumbered]{algorithm2e}
\usepackage{amsmath,amssymb,epsfig,verbatim}
\usepackage{soul}
\usepackage{lineno}
\let\doendproof\endproof
\renewcommand\endproof{~\hfill$\qed$\doendproof}
\begin{document}
\title{The well-separated pair decomposition for  balls} 
\author{ Abolfazl Poureidi  \and Mohammad Farshi}
\institute{Combinatorial and Geometric Algorithms Lab., \\ Department of Computer Science, Yazd University,  Yazd, Iran. \\ \email{a.poureidi@gmail.com, mfarshi@yazd.ac.ir}.}
\maketitle
\begin{abstract}
Given a real number $t>1$, a geometric \mbox{$t$-spanner} is a geometric graph for a point set in $\mathbb{R}^d$ with straight lines between vertices such that the ratio of the shortest-path distance between every pair of vertices in the graph (with Euclidean edge lengths) to their actual Euclidean distance is at most~$t$.
An imprecise point set is modeled by a set $R$ of regions in $\mathbb{R}^d$. If one chooses a point in each region of $R$, then the resulting point set is called a precise instance of~$R$.
An imprecise \mbox{$t$-spanner} for an imprecise point set $R$ is a graph $G=(R,E)$ such that for each precise instance $S$ of $R$, graph $G_S=(S,E_S)$, where $E_S$ is the set of edges corresponding to $E$, is a \mbox{$t$-spanner}.

In this paper, we show that, given a real number $t>1$, there is an imprecise point set $R$ of $n$ straight-line segments in the plane such that any imprecise \mbox{$t$-spanner} for $R$ has $\Omega(n^2)$ edges. Then, we propose an algorithm that computes a Well-Separated Pair Decomposition (WSPD) of size ${\cal O}(n)$  for a set of $n$ pairwise disjoint  $d$-dimensional balls with arbitrary sizes. Given a real number $t>1$ and given a set of $n$ pairwise disjoint $d$-balls with arbitrary sizes,  we use this WSPD  to compute in ${\cal O}(n\log n+n/(t-1)^d)$ time   an imprecise \mbox{$t$-spanner} with ${\cal O}(n/(t-1)^d)$ edges  for balls.
\\[5mm]
 \end{abstract}
 \section{Introduction}\label{sec.intro}
 We use a geometric algorithm to solve a geometric problem. The input of geometric problems is some spatial objects, for example, a set of points in the plane. In many problems of computational geometry, there exists an assumption that the input data is precise and known exactly. However, there are many aspects of uncertainty in data, such as input data have been collected using measuring equipment that is not precise enough, or may have been stored as floating point with a limited number of decimals. There are many work in computational geometry that consider geometric algorithms for imprecise inputs \cite{blmm-pipfdtsae-11,hm-ticpps-08,klm-pipast-10,lm-uiopipffold-14,lk-laschfip-10,lk-lbbsdarpoip-10}. In these work, each point is modeled by a region in~$\mathbb{R}^d$, and then for these regions constructing a geometric structure such as the convex hull, the Voronoi diagram, or the (Delaunay) triangulation is considered.
 
 A geometric network is a weighted undirected graph whose vertices are points in~$\mathbb{R}^d$, and in which each edge is a straight-line segment with weight equal to the Euclidean distance between its endpoints. In a geometric network $G = (P,E)$ on a set $P$ of $n$ points, the graph distance $d_G(u, v)$ of $u,v\in P$ is the length of the shortest path between $u$ and $v$ in~$G$. Then, $\delta_G(u,v)=\frac{d_G(u,v)}{|uv|}$ denotes the dilation between $u$ and $v$ in $G$. We say that there exists a \mbox{$t$-path} ($t>1$) between two vertices $u,v\in P$ in $G$ if $\delta_G(u,v)\leq t$ and a network $G$ is called a \mbox{$t$-spanner} if~$\delta_G(u,v)\leq t$ for any pair of distinct points $u,v\in P$. 
  
 We call any set $R=\{R_1,\ldots,R_n\}$ of $n$ regions in $\mathbb{R}^d$ an imprecise point set. For a given imprecise point set $R$, any set $S=\{p_1,\ldots,p_n\}$, where $p_i\in R_i$, for all $1\leq i\leq n$, is called a precise instance of~$R$. For a given imprecise point set $R$, a graph $G = (R,E)$, where $E$ is a set of unordered pairs of regions in $R$, is called an imprecise geometric graph. 
 
 Given an imprecise geometric graph $G = (R,E)$, and for each precise instance $S$ of $R$, we call the geometric graph $G_S = (S,E_S)$, where $E_S=\{(p_i,p_j)|(R_i,R_j)\linebreak \in{E}\}$, a precise instance of $G$ corresponding to~$S$. Also, we call $G$ an imprecise \mbox{$t$-spanner} ($t>1$), if $G_S$, for any precise instance $S$ of $R$, is a \mbox{$t$-spanner}. It is easy to see that if there are two overlapping regions in $R$, then there must be an edge between the overlapping regions in any $t$-spanner for $R$. Therefore, the number of edges of a $t$-spanner for $R$ depends on the number of overlapping regions. Hence, in the rest of the paper, we assume that $R$ contains only pairwise disjoint regions.
  
 Abam et al.~\cite{acfs-otpotsspd-13} considered the problem of constructing a spanner for $n$ pairwise disjoint balls in~$\mathbb{R}^d$. For a given $t>1$, they showed that there exists an imprecise \mbox{$t$-spanner} with ${\cal O}(n/(t-1)^d)$ edges that can be computed in ${\cal O}(n \log n+n/(t-1)^d)$ time when all balls have similar sizes. Their spanner construction was based on the Well-Separated Pair Decomposition (WSPD)~\cite{callahan1995decomposition} approach, see also Chapter 9 of the book by Narasimhan and Smid~\cite{ns-gsn-07b}. They obtained a WSPD of imprecise points, i.e., balls, using a WSPD of the center points. A Well-Separated Pair Decomposition (WSPD) for a point set $S\subset\mathbb{R}^d$ with respect to a real number $s>0$ is a set $\{\{A_i,B_i\}_i\}$ of pairs where (i) $A_i,B_i\subset S$, (ii) $A_i$ and $B_i$ are $s$-well-separated, i.e., there are $d$-dimensional balls $D_{A_i}$ and $D_{B_i}$ containing $A_i$ and $B_i$, respectively, such that $d(D_{A_i},D_{B_i})\geq s\times \max(\mbox{radius}(D_{A_i}),\mbox{radius}(D_{B_i}))$, and (iii) for any two points $p,q\in S$ there is exactly one index $i$ such that $p\in A_i$ and $q\in B_i$ or vice versa. When the sizes of the balls vary greatly, i.e. there is a set of $n$ pairwise disjoint balls in~$\mathbb{R}^d$ with arbitrary sizes, they used a Semi-Separated Pair Decomposition (SSPD) \cite{abam2009region,varadarajan1998divide} to solve the problem. They proved that there is an imprecise \mbox{$t$-spanner} with ${\cal O}(n \log n/(t-1)^{2d})$ edges that can be computed in ${\cal O}(n \log n/(t-1)^{2d})$ time. They constructed an SSPD of imprecise points using an SSPD of the center points. An SSPD is defined as a WSPD, except that, instead of $A_i$ and $B_i$ are $s$-well-separated in the condition (ii) we have $A_i$ and $B_i$ are $s$-semi-separated, i.e., there are balls $D_{A_i}$ and $D_{B_i}$ containing $A_i$ and $B_i$, respectively, such that $d(D_{A_i},D_{B_i})\geq s\times \min(\mbox{radius}(D_{A_i}),\mbox{radius}(D_{B_i}))$.
 
  Zeng and Gaoy \cite{zeng2014linear} considered the construction of a Euclidean spanner for $n$ balls in~$\mathbb{R}^d$ with radius $r$ in two phases. In the first phase, they preprocessed balls in time ${\cal O}(n(r+1/\varepsilon)^d\log\alpha)$, where $\alpha$ is the ratio between the farthest and the closest pair of centers of the balls. In the second phase, they could compute (or update) a $(1 +\varepsilon)$-spanner for any precise instance of the balls with ${\cal O}(n(r+1/\varepsilon)^d)$ edges in time ${\cal O}(n(r + 1/\varepsilon)^d\log(r+1/\varepsilon))$.
 
 In this paper, we consider the problem of computing an imprecise $t$-spanner for $n$ pairwise disjoint balls in~$\mathbb{R}^d$, given a real number $t>1$. These balls have arbitrary sizes. We present an algorithm that computes an imprecise $t$-spanner with ${\cal O}(n)$ edges in ${\cal O}(n\log n)$ time, when $t$ and $d$ are constants. The algorithm uses the WSPD to compute this imprecise spanner. Also, we give a set of pairwise disjoint regions in the plane such that any imprecise $t$-spanner for the regions is the complete graph.
 
 The organization of the paper is as follows. In Section~\ref{sec:spanner-n2-edge}, we prove that there is a set of $n$ pairwise disjoint straight-line segments in the plane such that any imprecise \mbox{$t$-spanner} for the segments has $\Omega(n^2)$ edges. Then, given pairwise disjoint balls in~$\mathbb{R}^d$ with arbitrary sizes, and given a real number $t>1$, we consider the problem of computing an imprecise $t$-spanner for the balls. In Section~\ref{sec:wspd-different-disk2}, we use the WSPD to compute an imprecise $t$-spanner for the balls with ${\cal O}(n/(t-1)^d)$ edges in ${\cal O}(n\log n+n/(t-1)^d)$ time.
 
 \section{An imprecise spanner with quadratic size}\label{sec:spanner-n2-edge}
 In this section, we present a set of pairwise disjoint convex regions in the plane such that any imprecise $t$-spanner for the regions, for any given $t>1$, must be the complete graph. This shows that it is not interesting to study imprecise spanners for any set of regions.
 
 Let $n\geq2$ be an integer, and define $\theta:=2\pi/n$. If we rotate the positive $x$-axis by angles $i\theta$, for each $i$ with $0\le i <n$, then we get $n$ rays. We number the rays starting from the positive $x$-axis and in counter-clockwise order. We denote the set of all these rays by~$\mathcal{R}^o_n$.
 
 Let us model an imprecise point as a line segment, and let $O^t_n$ be a set of pairwise disjoint line segments in the plane that is constructed as follows. Let $D_1$ and $D_2$ be two disks centered at the origin that have radii $0.4$ and $(t+1)/2$, respectively. Let $p_i$ and $q_i$, for $0\leq i<n$, be the intersections of $i$-th ray in $\mathcal{R}^o_n$ with the boundaries of $D_1$ and $D_2$, respectively. The line segment joining $p_i$ and $q_i$, denoted by $(p_i,q_i)$, is an element of $O^t_n$, see Figure~\ref{fig-stainer}. It is easy to see that $|p_iq_i|> t/2$.
 
 \begin{figure}
 \begin{center} 
 \includegraphics[width=0.35\textwidth]{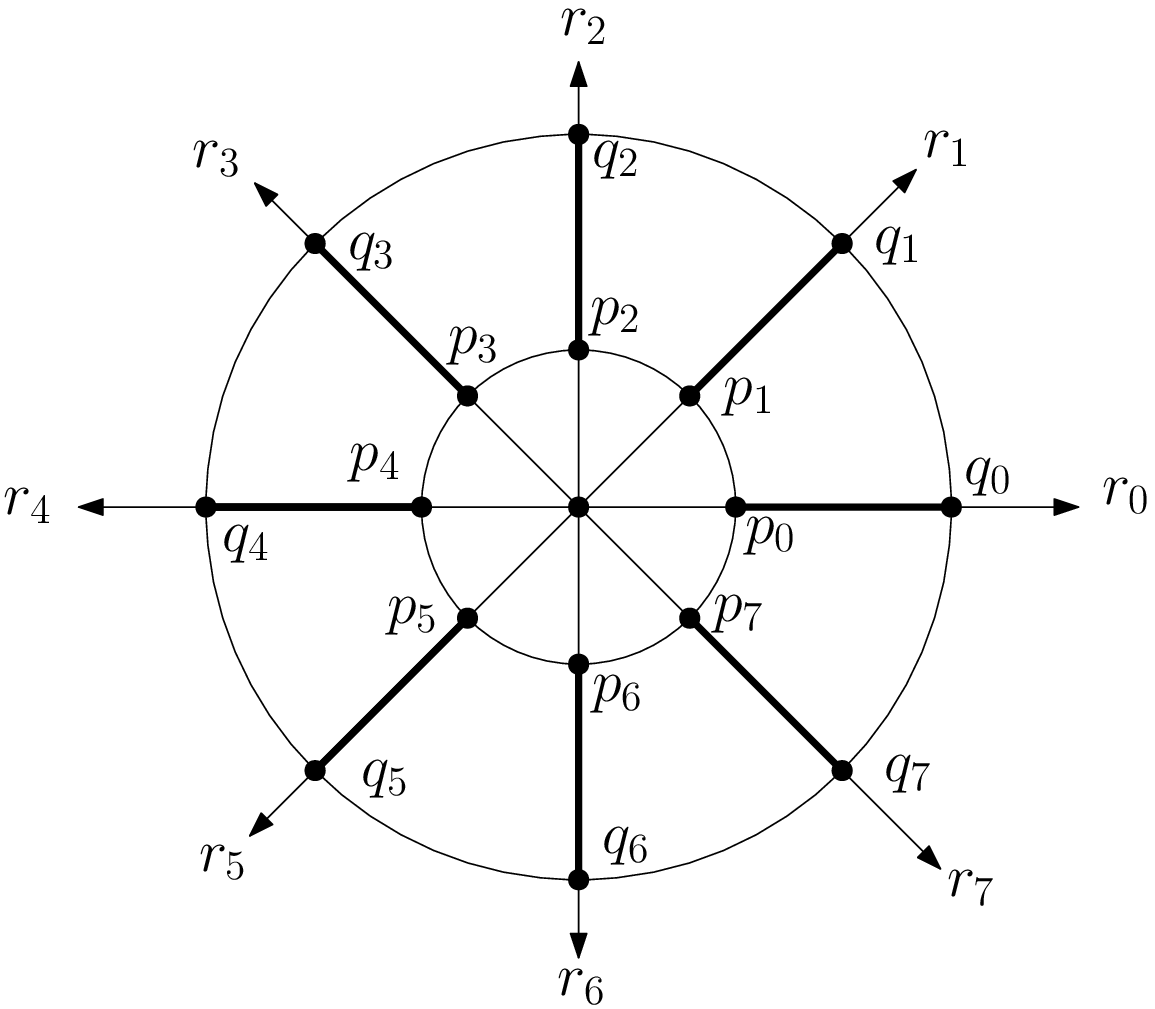}
 \end{center}
 \caption{Illustrating $O^t_8$.}
 \label{fig-stainer}
 \end{figure}
 
 \begin{lemma}\label{lem:n2-edge}
 The complete graph is the only imprecise spanner of $O^t_n$, for any $t>1$.
 \end{lemma}
 \begin{proof}
 Assume that $(p_i,q_i)$ and $(p_j,q_j)$ are two distinct line segments in $O^t_n$, where $0\leq i<j< n-1$. Let $G$ be an imprecise $t$-spanner for $O^t_n$ with no edge between $(p_i,q_i)$ and $(p_j,q_j)$.  
 Consider the precise instance $S=\{q_0,\ldots,q_{i-1},p_i,\linebreak q_{i+1},\ldots,q_{j-1},p_j,q_{j+1},\ldots,q_{n-1}\}$ of $O^t_n$, that is, choose $p_i$ and $p_j$ on $(p_i,q_i)$ and $(p_j,q_j)$, respectively, and $q_k$ on other line segments, for $0\leq k< n-1$ with $k\not=i,j$. It is clear that $|p_ip_j|<1$. Since there is no edge between $p_i$ and $p_j$ in $G_S$, the shortest path between $p_i$ and $p_j$ in $G_S$ passes through some $q_k$, for some $0\leq k< n-1$ with $k\not=i,j$. The Euclidean distance between $p_i$ and $q_k$ and the Euclidean distance between $p_j$ and $q_k$ are greater than $t/2$ and, hence, it follows that $d_{G_S}(p_i,p_j)\geq |p_iq_k|+|q_kp_j|> t$. Therefore, we get $\delta_{G_S}(p_i,p_j)> t$, which is a contradiction, because we assume that $G$ is an imprecise $t$-spanner for $O^t_n$. Hence, there must be an edge between any two distinct elements of $O^t_n$ in any imprecise $t$-spanner for~$O^t_n$.
 \end{proof}
 
 \section{An imprecise spanner for balls}\label{sec:wspd-different-disk2}
 Let $D=\{D_1,\ldots,D_n\}$ be a set of $n$ pairwise disjoint $d$-dimensional balls. In this section, we present an algorithm that computes an imprecise spanner for $D$ with ${\cal O}(n)$ edges in ${\cal O}(n\log n)$ time. The algorithm uses the WSPD \cite{callahan1995decomposition,ns-gsn-07b} for computing the imprecise spanner.
 
 \subsection{A well-separated pair for balls}\label{sec:def-wspd-ip}
 Let $X$ be a bounded point set of $\mathbb{R}^d$. We define bounding box of $X$, denoted by $R(X)$, as the smallest axes-parallel $d$-dimensional hyperrectangle that contains $X$. A $d$-dimensional hyperrectangle $R$ is the Cartesian product of $d$ closed intervals. More formally, $$R=[l_1,r_1]\times[l_2,r_2]\times\ldots\times[l_d,r_d],$$
 where $l_i$ and $r_i$ are real numbers with $l_i\leq r_i$, for $1\leq i\leq d$. We denote the length of $R$ in the $i$-th dimension by $L_i(R)=r_i-l_i$. We denote the maximum and minimum lengths of $R$ by $L_{\max}(R)$ and $L_{\min}(R)$, respectively. Let $C_X$ be a $d$-dimensional ball that contains $R(X)$. We denote the distance between two disjoint $d$-dimensional balls $C$ and $C'$ by $d(C,C')$, i.e., $$d(C,C')=|cc'|-(r+r'),$$
 where $c$ and $r$ are the center and radius, respectively, of $C$, and $c'$ and $r'$ are the center and radius, respectively, of $C'$. (Clearly, if $C$ or $C'$ is a point, then its radius is zero.)
 
 \begin{definition}\cite{callahan1995decomposition,ns-gsn-07b}\label{def:well-seperated-pair-on-points}
 Let $s>0$ be a real number, and let $A$ and $B$ be two finite sets of points in $\mathbb{R}^d$. We say that $A$ and $B$ are well-separated with respect to $s$ (or $s$-well-separated) if there are two disjoint $d$-dimensional balls $C_A$ and $C_B$, such that
 \begin{enumerate}
 \item
 $C_A$ and $C_B$ have the same radius, and
 \item
 $d(C_A,C_B)\geq s\times \mbox{radius}(C_A)$.
 \end{enumerate}
 \end{definition}
 In the following, we define $s$-well-separated for sets $A$ and $B$ of balls. Assume that $A$ or $B$ contains at least one nondegenerate ball, i.e., a ball with a positive radius. Let $D=\{D_1,\ldots,D_n\}$ be a set of $n$ pairwise disjoint $d$-dimensional balls with arbitrary sizes, and let $c_i$ be the center of $D_i$, for all $1\leq i\leq n$. For any $A\subseteq D$, let $A'=\{c_i|D_i\in A\}$.
 \begin{definition}\label{def:well-seperated-pair}
 Let $s>0$ be a real number, and let $A$ and $B$ be two nonempty subsets of~$D$. We say that $A$ and $B$ are well-separated with respect to $s$ (or $s$-well-separated) if there are two disjoint $d$-dimensional balls $C_{A'}$ and $C_{B'}$ with the same radius, such that one of the following conditions holds:
 \begin{itemize}
 \item[$\bullet$]
 $|A|=|B|=1$,
 \item[$\bullet$]
 $A=\{D_k\}$, for some $1\leq k\leq n$, $|B|>1$, and $d(c_k,C_{B'})-r_k\geq (3s+4)\times \mbox{radius}(C_{B'}),$
 \item[$\bullet$]
 $|A|>1$, $B=\{D_k\}$ for some $1\leq k\leq n$, and $d(c_k,C_{A'})-r_k\geq (3s+4)\times \mbox{radius}(C_{A'}),$ or 
 \item[$\bullet$]
 $|A|>1$, $|B|>1$, and $d(C_{A'},C_{B'})\geq (3s+4)\times \mbox{radius}(C_{A'}).$
 \end{itemize}
 \end{definition}
 
 It is easy to see that if all balls of $A$ and $B$ are degenerate (balls with radius $0$) and $A$ and $B$ are well-separated with respect to $s$ by Definition~\ref{def:well-seperated-pair}, then $A$ and $B$ are well-separated with respect to $s$ by Definition~\ref{def:well-seperated-pair-on-points}, too. In the rest of the paper, we accept the following convention. Let $A$ and $B$ be $s$-well-separated. If both $A$ and $B$ contain only points of $\mathbb{R}^d$, then $A$ and $B$ are $s$-well-separated by Definition~\ref{def:well-seperated-pair-on-points}. If $A$ or $B$ contains at least one nondegenerate ball, then $A$ and $B$ are $s$-well-separated by Definition~\ref{def:well-seperated-pair}. Let $S=\{p_1,\ldots,p_n\}$, where $p_i\in D_i$ for each $i$ with $1\leq i\leq n$, be a precise instance of $D$, and for any $A\subseteq D$, let $A_S=\{p_i\in S|D_i\in A \}$.
 
 \begin{lemma}\label{lem:wspd-for-D}
 Let $A$ and $B$ be two nonempty subsets of $D$ that are well-separated with respect to $s$, where $s>0$ is a real number and $A$ or $B$ contains at least one nondegenerate ball. Let $S=\{p_1,\ldots,p_n\}$ be an arbitrary precise instance of $D$, where $p_i\in D_i$ for all $1\leq i\leq n$. Then, $A_S$ and $B_S$ are $s$-well-separated.
 \end{lemma}
 
 \begin{proof}
Recall that for any $A\subseteq D$, we have $A'=\{c_i|D_i\in A\}$, where $c_i$ is the center of $D_i$. Since $A$ and $B$ are $s$-well-separated, by Definition~\ref{def:well-seperated-pair}, there are disjoint $d$-dimensional balls $C_{A'}$ and $C_{B'}$ with the same radius, such that one of the following cases holds for $A$ and $B$. In each case, we prove that $A_S$ and $B_S$ are $s$-well-separated, by Definition~\ref{def:well-seperated-pair-on-points}. 
 \begin{itemize}
 \item[$\bullet$]
 $|A|=|B|=1$. 
 
 Since both $A$ and $B$ are singletons, it is clear that $A_S$ and $B_S$ are $s$-well-separated.
 \item[$\bullet$] 
 $A=\{D_k\}$ for some $1\leq k\leq n$, $|B|>1$ and $d(c_k,C_{B'})-r_k\geq (3s+4)\times \mbox{radius}(C_{B'})$. 
 
 Let $\rho:=\mbox{radius}(C_{B'})$, and let $C_B$ be a $d$-dimensional ball with radius $3\rho$ co-centered with $C_{B'}$. Since $|B|>1$, the radius of each ball in $B$ is at most $2\rho$. (If $B$ contains a ball with the radius greater than $2\rho$, then $B$ is a singleton, contradicting our assumption that $|B|>1$.) So, $C_B$ contains all balls in $B$. Also, it is easy to see that $C_B$ contains bounding box $R(B_S)$. Therefore, 
 \begin{eqnarray*}
 d(c_k,C_B)-r_k&=&d(c_k,C_{B'})-2\rho-r_k\\
 &\geq& (3s+4)\times\rho-2\rho\\
 &=&(3s+2)\times\rho.
 \end{eqnarray*}
 Consider a $d$-dimensional ball $C_{A_S}$ with radius $3\rho$ that is centered at a point on the line passing through $p_k(\in S)$ and the center of $C_B$, such that $p_k$ is on the boundary of $C_{A_S}$ and $p_k$ is between the centers of $C_{A_S}$ and $C_B$. See Figure~\ref{fig-wsp-a-b2}. Since $A=\{D_k\}$ and $C_{A_S}$ contains $p_k(\in D_k)$,  ball $C_{A_S}$ contains bounding box $R(A_S)$.
 It follows that 
 \begin{eqnarray*} 
 d(C_{A_S},C_B)&\geq& d(D_k,C_B)\\
 &=& d(c_k,C_B)-r_k\\
  &\geq&(3s+2)\rho\\
 &\geq& s\times(3\rho).
 \end{eqnarray*} 
 So, there are $d$-balls $C_{A_S}$ and $C_B$ with radii $3\rho$ containing $R(A_S)$ and $R(B_S)$, respectively, such that $d(C_{A_S},C_B)\geq s\times(3\rho)$. It follows that $A_S$ and $B_S$ are $s$-well-separated.
 \begin{figure}
 \begin{center} 
 \includegraphics[width=0.6\textwidth]{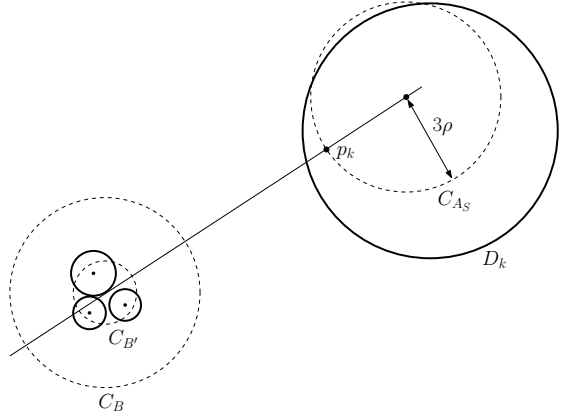}
 \end{center}
 \caption{Illustrating $C_{A_S}$ for $A=\{D_k\}$ and $B$, where $|B|>1$, in the plane for the second case of Lemma~\ref{lem:wspd-for-D}.}
 \label{fig-wsp-a-b2}
 \end{figure}
 
 \item[$\bullet$]
 
 $|A|>1$, $B=\{D_k\}$ for some $1\leq k\leq n$, and $d(c_k,C_{A'})-r_k\geq s'\times \mbox{radius}(C_{A'})$.
 
 The proof is similar to the previous case.
 \item[$\bullet$]
 $|A|>1$, $|B|>1$, and $d(C_{A'},C_{B'})\geq (3s+4)\times \mbox{radius}(C_{A'})$.
 
 Let $\rho:=\mbox{radius}(C_{A'})=\mbox{radius}(C_{B'})$, and let $C_A$ and $C_B$ be two $d$-dimensional balls with radii $3\rho$ co-centered with $C_{A'}$ and $C_{B'}$, respectively. Hence, $C_A$ contains bounding box $R(A_S)$ and $C_B$ contains bounding box $R(B_S)$. We get
 \begin{eqnarray*} 
 d(C_A,C_B)&=& d(C_{A'},C_{B'})-4\rho\\
  &\geq&(3s+4)\rho-4\rho\\
 &=&s\times(3\rho).
 \end{eqnarray*} 
 Therefore, $A_S$ and $B_S$ are $s$-well-separated.
 \end{itemize}
 So, we prove that if $A$ and $B$ are $s$-well-separated, then $A_S$ and $B_S$ are $s$-well-separated.
 \end{proof}
 
 \subsection{The WSPD for balls}
 
 Recall that $D=\{D_1,\ldots,D_n\}$ is a set of $n$ pairwise disjoint $d$-dimensional balls with arbitrary sizes.
 \begin{definition}\label{def:wspd}(Well-Separated Pair Decomposition of balls).
  Let $s>0$ be a real number. A well-separated pair decomposition (WSPD) for $D$, with respect to $s$, is a set $$\{\{A_1,B_1\},\{A_2,B_2\},\ldots,\{A_m,B_m\}\}$$ of pairs of nonempty subsets of $D$, for some integer $m$, such that
 \begin{enumerate}
 \item
 for any $i$ with $1\leq i\leq m$, $A_i$ and $B_i$ are $s$-well-separated (by Definition~\ref{def:well-seperated-pair}), and
 \item
 for any two distinct balls $D_p$ and $D_q$ of $D$, where $1\leq p,q\leq n$, there is a unique index $i$ with $1\leq i\leq m$, such that
 \begin{itemize}
 \item[$\bullet$]
 $D_p \in A_i$ and $D_q\in B_i$, or 
 \item[$\bullet$]
 $D_p \in B_i$ and $D_q\in A_i$.
 \end{itemize}
 \end{enumerate}
 \end{definition}
 
 We call $m$ as the size of the WSPD. Recall that if $S=\{p_1,\ldots,p_n\}$ is an arbitrary precise instance of $D$, then for any $A\subseteq D$, we have $A_S=\{p_i\in S|D_i\in A\}$.
 
 \begin{lemma}\label{lem:wspd-to-spanner}
 Let $s>0$ be a real number, and let $S=\{p_1,\ldots,p_n\}$ be an arbitrary precise instance of $D$. If $\{\{A_i,B_i\}|1\leq i\leq m\}$ is a WSPD for $D$ with respect to $s$, then $\{\{A_{i_S},B_{i_S}\}|1\leq i\leq m\}$ is a WSPD for $S=\{p_1,\ldots,p_n\}$ with respect to $s$.
 \end{lemma}
 \begin{proof}
 By Lemma~\ref{lem:wspd-for-D}, the proof is straightforward.
 \end{proof}
 
 If we can compute a WSPD for $D$, then (by Lemma~\ref{lem:wspd-to-spanner}) we can compute a WSPD for any precise instance of $D$. Callahan and Kosaraju \cite{callahan1995decomposition} used the split tree to compute a WSPD for a point set in $\mathbb{R}^d$. We also use the split tree to compute a WSPD for $D$. 
 \begin{figure}
 \begin{center} 
 \includegraphics[width=0.45\textwidth]{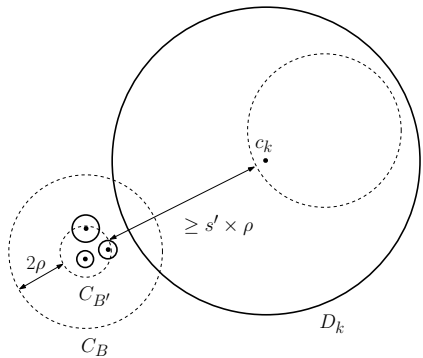}
 \end{center}
 \caption{$A'=\{c_k\}$ and $B'$, where $|B'|>1$, are well-separated with respect to \mbox{$s':=3s+6$}, but $A=\{D_k\}$ and $B$ are not $s$-well-separated.}
 \label{fig-n-wsp-a-b}
 \end{figure}

 To compute a WSPD of $D$, we construct a split tree $T$ on centers of all balls in~$D$. Then, we construct  a WSPD $W'$ of the centers with respect to $3s+6$ using~$T$. Next, we transform $W'$ to  a WSPD of  $D$, denoted by $W$, in the following way. For each pair $\{A',B'\}$  in $W'$, if both $A'$ and $B'$ are singletons  or both $A'$ and $B'$ contain more than one element, then we add  $\{A,B\}$ to~$W$, where $X$ is the set of all balls in that their centers are in $X'$. 
 Note that, by Definition~\ref{def:well-seperated-pair}, $A$ and $B$ are well-separated with respect to $s$. 
 Otherwise, one of sets $A'$ and $B'$ is a singleton and the other one contains more than one element.
 In this case, it is possible that $A$ and $B$ are not $s$-well-separated, see Figure~\ref{fig-n-wsp-a-b}. Without loss of generality, we assume that $|A'|=1$ and $|B'|>1$.  We check pair $\{A,B\}$ to see if it is a $s$-well-separated pair (by Definition \ref{def:well-seperated-pair}). If it is a $s$-well-separated pair, then we add it to $W$ and otherwise we partition $B'$ to $\{B'_i\}_i$ such that $\{A, B_i\}$ are $s$-well-separated pairs and then add them to~$W$. For the details of the algorithm, see Algorithm~\ref{alg:WSPD}.
 
 \begin{algorithm}
 \SetAlgoLined
 \KwIn{ $D=\{D_1,\ldots,D_n\}$ is a set of $n$ balls in $R^d$ with arbitrary sizes and $s$ is a positive real number.}
 \KwOut{A  well-separated pair decomposition of $D$  with respect to $s$.}
  $F:=\{c_1,\ldots,c_n\}$, where $c_i$ is the center of  $D_i$\;
  $T:=$ a split tree of $F$\;
  $W':=$ a WSPD of  $F$ with respect to $3s+6$ using $T$\;
  $W:=\emptyset$\;
  \ForEach{$\{A',B'\}\in W'$}
  {
  \uIf{ $(|A'|=1\wedge|B'|=1)\vee(|A'|>1\wedge|B'|>1)$}
          {
          Add $\{A,B\}$ to $W$\;
          }
    \Else
          {
 		 {\tcc{assume $|A'|=1$ and $|B'|>1$}
 		 $v:=$ the leaf in $T$ corresponding to~$A'$\; 
 		 $w:=$ the node in $T$ corresponding to~$B'$\;
 		 Add pairs generated by {\sc FindPairs}$(T,v,w)$ to $W$\;
 		 }
          }
  }
  		 \Return $W$\;
  \caption{{\sc ComputeWSPD}$(D,s)$\label{alg:WSPD}}
 \end{algorithm}
 \medskip
In the following, we explain the details of the way of partitioning~$B'$. We know that $T$ is a split tree on the centers of all balls in $D$. For any node $u$ of $T$, let $S_u$ be the set of all points that are stored in the subtree of $u$. Let $\{A',B'\}$ be a pair of $W'$ such that $A'=\{c_k\}$, for some $1\leq k\leq n$, and $|B'|>1$. Assume that $v$ and $w$ are the nodes of $T$ such that $S_v=A'$ and $S_w=B'$. Obviously, $v$ is a leaf and $w$ is an internal node of $T$. 
Note that for each node $w$ in the split tree, the bounding box of $S_w$, denoted by $R(w)$, is stored at $w$. So, we can test in ${\cal O}(1)$ time whether there is a ball $C_{B'}$ containing $B'$ such that $d(c_k,C_{B'})-r_k\geq (3s+4)\times \mbox{radius}(C_{B'})$. To this end, let  $C_{B'}$ be the $d$-dimensional ball of radius $(\sqrt{d}/2)\times L_{\max}R(w)$ centered at the center of $R(w)$, where  $L_{\max}R(w)$ is the length of the longest side of $R(w)$ and the center of $R(w)$ is the intersection of perpendicular bisecting hyperplanes of sides of $R(w)$. If $d(c_k,C_{B'})-r_k\geq (3s+4)\times \mbox{radius}(C_{B'})$, then $\{A,B\}$ is a $s$-well-separated pair and so we add $\{A,B\}$ to~$W$. Otherwise,  we follow the above process by $\{v,w_l\}$ and $\{v,w_r\}$, where $w_l$ and $w_r$ are the left and the right children of $w$, respectively.

 For details of the partition algorithm, denoted by {\sc FindPairs}$(T,v,w)$,  see algorithm~\ref{alg:FINDPAIRS}. We may assume without loss of generality that always $|S_v|=1$, that is, $v$ is a leaf of $T$. Clearly, the algorithm {\sc FindPairs}$(T,v,w)$ terminates.
 
 \begin{algorithm}
 \SetAlgoLined
 \KwIn{An split tree $T$  and a pair $\{v,w\}$, where $v$ is a leaf and $w$ is an internal node of the split tree $T$.
 }
 \KwOut{A collection of well-separated pairs $\{A,B\}$ with respect to $s$, where $A'=S_v=\{c_k\}$ and $B'\subseteq S_w$.}
  \If{there is a ball $C_{S_w}$ such that $d(c_k,C_{S_w})-r_k\geq (3s+4)\times\mbox{radius}(C_{S_w})$}
          {
          \Return the pair $\{A,B\}$, where $A'=S_v$ and $B'=S_w$\;
          }
             $w_l:=$ left child of $w$\;
             $w_r:=$ right child of $w$\;
             {\sc FindPairs}$(T,v,w_l)$\;
             {\sc FindPairs}$(T,v,w_r)$\;
 \caption{{\sc FindPairs}$(T,v,w)$\label{alg:FINDPAIRS}}
 \end{algorithm}
 \medskip
  
  Now, we show that the algorithm generates a WSPD of $D$ with ${\cal O}(n)$ pairs.
 \begin{lemma}\label{lem:r_C_A_GT_r_C_B}
 If $A'=\{c_k\}$, for some integer $1\leq k\leq m$, and $B'$, where $|B'|>1$, are well-separated with respect to $3s+6$ (by Definition \ref{def:well-seperated-pair-on-points}), but $A$ and $B$ are not $s$-well-separated (by Definition \ref{def:well-seperated-pair}), then $r_k=\mbox{radius}(D_k)>\sqrt{d}\times L_{\max}(R(B'))$.
 \end{lemma}
 \begin{proof}
 The proof is by contradiction. Assume that $r_k\leq \sqrt{d}\times L_{\max}(R(B'))$. Let $\rho:=(\sqrt{d}/2)\times L_{\max}(R(B'))$. Since $A'$ and $B'$ are well-separated with respect to $3s+6$ (by Definition \ref{def:well-seperated-pair-on-points}), we have $d(c_k,C_{B'})\geq (3s+6)\times\rho$, where $C_{B'}$ is a $d$-ball with radius $\rho$ that is centered at the center of $R(B')$. So, $$d(c_k,C_{B'})-r_k\geq (3s+4)\times\rho.$$
 Therefore, $A$ and $B$ are well-separated with respect to $s$ (by Definition \ref{def:well-seperated-pair}), a contradiction.~\end{proof}
 
 \begin{lemma}\label{lem:W-is-WSPD}
 Set $W$ is a WSPD for $D$ with respect to $s$.
 \end{lemma}
 \begin{proof}
 It is easy to see that for all $\{A,B\}\in W$, sets $A$ and $B$ are $s$-well-separated (by Definition \ref{def:well-seperated-pair}). By \cite{callahan1995decomposition,ns-gsn-07b}, the proof of the second condition in Definition~\ref{def:wspd} is straightforward.
 \end{proof}
 It remains to prove an upper bound on $|W|$.  We can partition the pairs in $W$ into two categories. In the first category, there are pairs $\{A,B\}$ such that $\{A',B'\}$ is in~$W'$. Since the size of $W'$ is linear in $n=|D|$, obviously the number of pairs in this category is linear in $n$. The second category contains the pairs that generated by partitioning the sets in pairs of $W'$. In the following lemma, we show that the number of pairs in this category is also linear in $n$. To this end, we show that any set $B$ appears in at most a constant number of pairs in this category. Note that each pair in this category contains a singleton and a set that may contain more than one element.
  
 Let $Z$ be the set of all pairs of $W'$ that {\sc FindPairs} returns at least two pairs. More precisely, let $$Z=\{\{A'_i,B'_i\}| 1\leq i\leq q, \{A'_i,B'_i\}\in W',|A'_i|=1,|B'_i|>1\},$$ such that $S_v=A'_k$, for some leaf $v$ of $T$, and $S_w=B'_k$, for some node $w$ of~$T$, and algorithm {\sc FindPairs}$(v,w)$ returns at least two pairs, for all $k$ between $1$ and~$q$. Let $\{A_k,B\}$ be some pair returned by algorithm {\sc FindPairs}$(v,w)$ such that $B'=S_u\subset S_w$, for some node $u$ of~$T$. 
 In the following, we apply a packing argument (similar to Lemma 9.4.3 of \cite[Chapter 9]{ns-gsn-07b}) to prove that each $B$ is involved in at most a constant number (dependent only on $s$ and $d$) of pairs in $W$. Let $\pi(u)$ be the parent of node $u$ of $T$, except for the root.
  
 \begin{lemma}\label{lem:constant-pair}
 Set $B$ involved in  at most ${(3s+9)^d\times\Gamma(d/2+1)}/{\pi^{d/2}}$ pairs in $W$, where $\Gamma$ denotes Euler's gamma-function.
 \end{lemma}
 \begin{proof}
 Let $u$ be a node of $T$ such that $S_u=B'$, and let $B'_p=S_{\pi(u)}$. Let $x$ be the center of bounding box $R(B'_p)$, and $\rho:=\sqrt{d}\times L_{\max}(R(B'_p))$. Without loss of generality,  we assume that $\{D_1,B\},\ldots,\{D_r,B\}$ are all pairs of $W$ that contain $B$. Since $W$ is a WSPD for $D$, clearly, balls $D_i$ for all $1\leq i\leq r$ are pairwise distinct and, therefore, are pairwise disjoint. Assume $c_i$ and $v_i$ are the center and the leaf of $T$ corresponding to~$D_i$, respectively.
 
 Let $C$ be a hypercube  centered at point $x$, where $x$ is the center of bounding box $R(B'_p)$, and with  side length $(3s+5)\times\rho$.  We have $C\cap D_i\not=\emptyset$, because if $C\cap D_i=\emptyset$, then 
 \begin{eqnarray*} 
 d(c_i,C_{B'_p})-r_i&=& d(D_i,C_{B'_p})\\
 &>&  \frac{1}{2}\times\mbox{side-length}(C)-\mbox{radius}(C_{B'_p})\\
 &=& (3s+4)\times(\rho/2),
 \end{eqnarray*} 
  where $C_{B'_p}$ is a ball with center $x$ and radius $\rho/2$. (Clearly, $C_{B'_p}$ contains $R(B'_p)$.) Hence, $\{D_i\}$ and $B_p$ are $s$-well-separated (by Definition \ref{def:well-seperated-pair}), which is a contradiction because if $\{D_i\}$ and $B_p$ are well-separated with respect to $s$, then  {\sc FindPairs}$(T,v_i,\pi(u))$ finishes and does not  run {\sc FindPairs}$(T,v_i,u)$.

 Since each element of $Z$ is a well-separated pair with respect to $(3s+6)$,  the pair $\{\{c_i\},B'_p\}$ is also a well-separated pair with respect to $(3s+6)$. Since $\{D_i\}$ and $B_p$ are not $s$-well-separated, by Lemma~\ref{lem:r_C_A_GT_r_C_B}, for each $i$,  $1\leq i\leq r$, we have $r_i=\mbox{radius}(D_i)>\sqrt{d}\times L_{\max}(B'_p)=\rho$.
 
 For each $i$, let $C_i$ be a $d$-dimensional ball with radius $\rho$ such that $D_i$ contains $C_i$ and $C\cap C_i\not=\emptyset$. Since the  balls $D_i$ are pairwise disjoint, the balls $C_i$ are also pairwise disjoint.
  
 Let $C'$ be a hypercube with sides of length $(3s+5)\rho+4\rho$ and with center $x$. The length of sides of $C'$ is the sum of the length of sides of $C$ and two times the diameter of~$C_i$. Therefore, $C'$ contains all balls $C_i$, for each $i$ with $1\leq i\leq r$. The volumes of $\bigcup^r_{i=1}C_i$ and $C'$ are $r\times(\pi^{d/2}/\Gamma(d/2+1))\times\rho^d$ and $((3s+9)\times\rho)^d$, respectively. (The volume of a ball with radius $r$ in $\mathbb{R}^d$ is $({\pi^{d/2}}/{\Gamma(d/2+1)})\times r^d$.) Therefore, we get $r\times(\pi^{d/2}/\Gamma(d/2+1))\times\rho^d\leq ((3s+9)\times\rho)^d$. It follows that $$r \leq{(3s+9)^d\times\Gamma(d/2+1)}/{\pi^{d/2}},$$ which completes the proof.
 \end{proof}
 Since $T$ has ${\cal O}(n)$ nodes, it follows from Lemma \ref{lem:constant-pair} that $|Z|={\cal O}(n)$. To sum-up, we have the following result.
 
 \begin{corollary}\label{coro:size_W=O(n)}
 The set $W$ contains at most ${\cal O}(n)$ pairs.
 \end{corollary}
 
 Lemma~\ref{lem:W-is-WSPD} and Corollary~\ref{coro:size_W=O(n)} immediately imply the following result.
 
 \begin{theorem}\label{theo:wspd-on-D}
 Let $D=\{D_1,\ldots,D_n\}$ be a set of $n$ $d$-dimensional pairwise disjoint balls with arbitrary sizes, and let $s>0$ be a real number. There is a WSPD for $D$ with respect to $s$ of size ${\cal O}({s^d\times\Gamma(d/2+1)}/{\pi^{d/2}}\times n)$. The WSPD can be computed in ${\cal O}({s^d\times\Gamma(d/2+1)}/{\pi^{d/2}}\times n\log n)$ time by an algorithm that uses ${\cal O}({s^d\times\Gamma(d/2+1)}/{\pi^{d/2}}\times n)$ space.
 \end{theorem}
 
 \begin{theorem}\label{theo:spanner-on-D}
 Let $D=\{D_1,\ldots,D_n\}$ be a set of $n$ pairwise disjoint balls in $\mathbb{R}^d$, and let $t>1$ be a real number. There is an imprecise $t$-spanner for $D$ with ${\cal O}(n/(t-1)^d)$ edges. This imprecise $t$-spanner can be computed in ${\cal O}(n\log n+n/(t-1)^d)$ time.
 \end{theorem}
 \begin{proof}
 Let $s=4(t+1)/(t-1)$ and, by Theorem~\ref{theo:wspd-on-D}, let $\{\{A_i,B_i\}|1\leq i\leq m\}$ be a WSPD for $D$ with respect to $s$ of size $m={\cal O}({s^d\times\Gamma(d/2+1)}/{\pi^{d/2}}\times n )$. Initialize $E=\emptyset$. For $1\leq i\leq m$, we add edge $\{D_j,D_k\}$ to $E$, where $D_j\in A_i$ and $D_k\in B_i$. Let $G=(D,E)$ be the resulting graph. By Theorem~\ref{theo:wspd-on-D}, $G$ can be computed in ${\cal O}(n\log n)$ time. Let $S=\{p_1,\ldots,p_n\}$ be an arbitrary precise instance of $D$. By Lemma~\ref{lem:wspd-to-spanner}, $\{\{A_{i_S},B_{i_S}\}|1\leq i\leq m\}$ is a WSPD for $S$ with respect to $s$. It follows from \cite{callahan1995decomposition} that $G_S=(S,E_S)$ is a $t$-spanner for $S$, that is, $G=(D,E)$ is an imprecise $t$-spanner for $D$.
 \end{proof}
 \section{Conclusions}
 Given a real number $t>1$, in this paper, we present a set of pairwise disjoint line segments in the plane that any imprecise $t$-spanner for the segments is the complete graph. This shows that studying imprecise spanners for some regions is not interesting. Then, we compute a  WSPD with respect to a given real number $s>0$ of size ${\cal O}(n)$ for a set of $n$ pairwise disjoint \mbox{$d$-dimensional} balls with arbitrary sizes in ${\cal O}(n\log n)$ time, when $s$ and $d$ are constants. This WSPD helps us to compute imprecise spanners with ${\cal O}(n)$  edges for a set of $n$ pairwise disjoint $d$-balls that have arbitrary sizes.
\section*{Acknowledgments}
The authors would like to thank the reviewer for his/her  helpful and constructive  comments that improved
the paper.

 \bibliographystyle{elsarticle-num}
 \bibliography{Bib-file}
 \end{document}